\documentclass[conference,letterpaper]{IEEEtran}
\usepackage[left=.635in,right=.635in,top=0.75in,bottom=1.05in]{geometry}

\usepackage[utf8]{inputenc} 
\usepackage[T1]{fontenc}
\usepackage{url}
\usepackage{ifthen}
\usepackage{cite}
\usepackage[cmex10]{amsmath}

\usepackage{cite}
\usepackage[cmex10]{amsmath}
\usepackage{amssymb}
\usepackage{stfloats}
\usepackage{amsthm} 
\usepackage{wrapfig}
\usepackage{amssymb}
\usepackage{latexsym}
\usepackage{bm, bbm} 
\usepackage{epsfig}
\usepackage{graphicx}
\usepackage{epsfig}
\usepackage{multicol}
\usepackage{multirow, array}
\usepackage{psfrag}
\usepackage{float}
\usepackage{cite}
\usepackage{subfigure}
\usepackage{color}
\usepackage{mathrsfs}
\usepackage{balance}

\usepackage{tkz-graph}
\usepackage{tikz}
\usetikzlibrary{shapes.geometric}
\usepackage[noend]{algpseudocode}

\usepackage{algorithm} 
\usepackage{algorithmicx} 
\usepackage{tabularx} 

\usepackage{mathtools}
\usepackage{blkarray}
\usepackage{bigstrut}

\let\labelindent\relax
\usepackage{enumitem}

\newtheorem{lemma}{Lemma}
\newtheorem{definition}[lemma]{Definition}
\newtheorem{theorem}[lemma]{Theorem}
\newtheorem{prop}[lemma]{Proposition}
\newtheorem{corollary}[lemma]{Corollary}

\theoremstyle{plain}

\theoremstyle{plain}

\newtheorem{PreRemark}[lemma]{{\textbf{Remark}}}
    {\begin{PreRemark}\upshape}{\hfill$\square$\end{PreRemark}}

\newtheorem{PreExample}[lemma]{{\textbf{Example}}}
    {\begin{PreExample}\upshape}{\hfill$\square$\end{PreExample}}

\long\def\symbolfootnote[#1]#2{\begingroup%
\def\thefootnote{\fnsymbol{footnote}}\footnote[#1]{#2}\endgroup} 

\makeatletter
\newcommand{\multiline}[1]{%
  \begin{tabularx}{\dimexpr\linewidth-\ALG@thistlm}[t]{@{}X@{}}
    #1
  \end{tabularx}
}
\makeatother

\MakeRobust{\Call}

\makeatletter
\renewcommand*\env@matrix[1][*\c@MaxMatrixCols c]{%
  \hskip -\arraycolsep
  \let\@ifnextchar\new@ifnextchar
  \array{#1}}
\makeatother

\allowdisplaybreaks[4]

\begin{document}
\title{Low-Density Parity-Check Codes of High Girth from Permutation Polynomials}

 \author{%
   \IEEEauthorblockN{Tasawwar Hussain, Anthony G\'omez-Fonseca, and Tefjol Pllaha}
   \IEEEauthorblockA{Department of Mathematics and Statistics, University of South Florida, Tampa, FL 33620, USA \\ \{thussain, agomezfonseca, tpllaha\}@usf.edu}
 }

\maketitle

\begin{abstract}
Quasi‑cyclic low‑density parity‑check (QC‑LDPC) codes are widely used in modern communication standards due to their simple structure and efficient encoder and decoder implementations. However, the commutativity of the circulant permutation matrices used in their construction significantly limits their girth and minimum distance, affecting iterative decoding performance. In this paper, we investigate protograph-based codes constructed using permutation polynomials of higher degree, focusing on a closed family of non-commuting binomial permutation polynomials that enables these constructions to exceed both of these girth and minimum distance upper bounds while maintaining a simple representation. Simulation results demonstrate improved decoding performance compared to constructions based on affine permutation polynomials, which include QC-LDPC codes.
\end{abstract}


\vspace{-1mm}
\section{Introduction}
\label{sec:introduction}

The attractiveness of a linear code for adoption in industry standards relies, usually, on the simplicity of its structure and the benefits that can be obtained by exploiting said structure, such as the possibility to be encoded with low complexity and the efficiencies in decoder design. The family of low-density parity-check (LDPC) codes, and in particular quasi-cyclic (QC) LDPC codes, has been widely used in industry standards precisely because of those benefits. In the case of QC-LDPC codes, a low-complexity encoding using simple feedback shift-registers \cite{lcz+06} and a memory-efficient partially-parallel decoder \cite{wc07} are among the reasons why they now appear in many industry standards.

The simple structure of QC-LDPC codes, characterized by the commutativity of the circulant permutation matrices used in their construction, is however known to be a limiting factor in many aspects. 
In particular, upper bounds on both the minimum distance \cite{md01,fossorier04} and the girth of their corresponding Tanner graphs \cite{sv12} can be derived for any protograph-based construction. 
The girth \cite{mb01}, together with the number of short cycles \cite{hc06}, and other graphical structures composed of short cycles, such as trapping sets \cite{kb12_b}, plays an important role in determining their performance under iterative decoding algorithms, particularly because of the introduction of correlation during decoding iterations \cite{ru01} and the emergence of error floors. As a consequence, it is desirable to rescind from the commutativity of the permutation matrices used in the construction of protograph-based codes while maintaining the attractiveness of a compact representation. 
In this paper, we pursue this idea by proposing the construction of protograph-based codes using {\em permutation polynomials of higher degree}. 

First studied by Hermite and Dickson \cite{h1863,d1896}, permutation polynomials have found many applications in mathematics, particularly in the theory of finite fields, with connections to coding theory and combinatorial designs, as well as to cryptography \cite{d13,dy06,lm84,ccz98,d99a,d99b}. 
Affine permutation polynomials have already been used to construct both classical and quantum protograph-based codes in \cite{myp06,yk19,gg24,k25}. Since the composition of affine permutation polynomials does not commute in general, the corresponding permutation matrices do not commute, which means that the upper bounds on the minimum distance and girth that limit the performance of QC-LDPC codes do not apply. 
Moreover, as exemplified in this paper, a family of permutation polynomials that is closed under composition and inverses allows for efficient girth and cycle analysis of protograph-based codes. 
In this setting, the girth conditions can be translated into explicit conditions imposed on the coefficients or exponents of the polynomials. The goal of this paper is to exploit this framework to construct and analyze protograph-based codes with improved minimum distance and girth properties, while retaining a simple representation amenable to efficient implementation.

This paper is structured as follows. In Section \ref{sec:background}, we introduce all relevant preliminaries and notations.
In Section \ref{sec:pp_beyond_affine}, we present protograph-based codes using monomial permutation polynomials and discuss their minimum distance and girth upper bounds. Also, we present the analysis of a family of non-commuting binomial permutation polynomials that is closed under composition and inverses. We use these polynomials to construct protograph-based codes with large girth and compare them to some constructions using affine permutation polynomials.
In Section \ref{sec:simulations}, we present some simulation results of selected codes, showing an improved performance compared to those constructed using affine permutation polynomials. Lastly, we conclude the paper in Section \ref{sec:conclusion}.

\vspace{-1mm}
\section{Background}
\label{sec:background}

\subsection{Protograph-based Codes}
\label{protograph_based_codes_background}
A protograph-based code $\mathcal{C}$ is constructed in the following way. Consider a relatively small bipartite graph $G_B$, known as the protograph \cite{tho03}, with biadjacency matrix $B={(b_{ij})}_{n_c\times n_v}$, where $b_{ij}$ is a nonnegative integer for $i\in[n_c],j\in[n_v]$, and where $[l]=\{0,1,\dots,l-1\}$. In this paper, we focus on the all-ones protograph for space constraints, so $b_{ij}=1$ for $i\in[n_c],j\in[n_v]$. From $B$, we construct a scalar parity-check matrix $H$ by replacing each $b_{ij}$ with an $N\times N$ permutation matrix $H_{ij}$ for $i\in[n_c],j\in[n_v]$. 

From the parity-check matrix $H$, we construct a bipartite graph $G=(V,E)$, called a Tanner graph \cite{tan81}, by considering $H$ as its biadjacency matrix.
This bipartite graph represents the protograph-based code $\mathcal{C}$ obtained from $H$. The set $V$ is the set of vertices (or nodes) and $E$ is the set of edges. Label the vertices of $G$ by $v_a$, for $a=0,1,2,\dots,|V|-1$, and the edges by $e_b$, for $b=0,1,2,\dots,|E|-1$. Each edge $e_b$ has the form $e_b=(v_a,v_c)$, for some $v_a,v_c\in V$, and the vertices $v_a$ and $v_c$ are called the endpoints of $e_b$. A (directed) walk $W$ of length $m$ in the graph $G$ is an alternating sequence $W=v_0e_1v_1e_2\cdots v_{m-1}e_mv_m$ of vertices and edges such that $e_l=(v_{l-1},v_l)\in E$ for all $1\leq l\leq m$. The first vertex appearing in the alternating sequence, $v_0$, is called the base point of $W$. A walk $W$ is said to be closed if the two endpoints are the same, this is, when $v_0=v_m$. A closed walk $W$ is backtrackless if $e_l\neq e_{l+1}$ for all $l=1,2,\dots,m-1$. A backtrackless closed walk $W$ is tailless if $e_m\neq e_1$, and $W$ is called, in this case, a TBC walk. A cycle is a closed walk $W$ having distinct vertices and distinct edges, and if its alternating sequence has $k$ edges in it, then we call $W$ a $k$-cycle. The length of a shortest cycle is called the \textit{girth} of the graph.

\subsection{Permutation Polynomials}
\label{subsec:permutation_polynomials}
Let $R\in\{\mathbb{Z}_N,\mathbb{F}_q\}$, where $N$ is a positive integer, and $q=p^n$ for some prime $p$ and a positive integer $n$. A polynomial $f\in R[x]$ is a permutation polynomial of $R$ if the function $f:c\rightarrow f(c)$ from $R$ into itself induces a permutation. 
Associated to $f$, there is a (binary) permutation matrix that has 1 in position $(c,f(c))$ and 0 elsewhere.
Any protograph-based code $\mathcal{C}$ can be described by a polynomial parity-check matrix. 
For the all-ones protograph $G_B$, the polynomial parity-check matrix is given by $H={\left[\begin{matrix}p_{ij}(x)\end{matrix}\right]}_{n_c\times n_v}$, where each $p_{ij}(x)$, $i\in[n_c],j\in[n_v]$, is a permutation polynomial. We obtain a scalar parity-check matrix, also denoted by $H$ since there will be no ambiguity, by replacing each permutation polynomial $p_{ij}(x)$ with its corresponding permutation matrix.

Although not proposed using the framework of permutation polynomials, the well-known family of quasi-cyclic (QC) LDPC codes provides an excellent example of protograph-based codes constructed using monic affine permutation polynomials. If $P$ is a circulant permutation matrix with $P_{0,a}=1$, $a\in[N]$, then the permutation polynomial corresponding to $P$ is $x+a$. In the case of the all-ones protograph, the parity-check matrix of a QC-LDPC code is given by $H={\left[\begin{matrix}x+a_{ij}\end{matrix}\right]}_{n_c\times n_v}$, where $a_{ij}\in[N]$ for $i\in[n_c],j\in[n_v]$.

A natural generalization of QC-LDPC codes in the context of permutation polynomials, which motivates the present work, gives the class of affine permutation matrices (APM) LDPC codes. First introduced in \cite{myp06}, APM-LDPC codes are constructed by lifting a protograph using linear polynomials of the form $f(x)=ax+b$. The linear polynomial $f(x)=ax+b$ is a permutation polynomial if and only if $\gcd(a,N)=1$. 
In the case of the all-ones protograph, the parity-check matrix of an APM-LDPC code is given by $H={\left[\begin{matrix}a_{ij}x+b_{ij}\end{matrix}\right]}_{n_c\times n_v}$, where $a_{ij},b_{ij}\in[N]$ and $\gcd(a_{ij},N)=1$ for $i\in[n_c],j\in[n_v]$.
Girth \cite{yk19} and cycle analysis \cite{gg24} have been presented for these codes.

\subsection{Girth of Protograph-based Codes}
\label{subsec:girth_conditions}
In \cite{fan00}, Fan presented algebraic conditions imposed by a cycle in the Tanner graph of a protograph-based LDPC code, constructed using the all-ones protograph, on the product of certain permutation matrices in the scalar parity-check matrix $H$. In the following theorem, we generalize this result to any protograph using the notion of TBC walks. 
\begin{theorem}
\label{cycles_and_fixed_points_scalar}
Let $\mathcal{C}$ be a protograph-based code with parity-check matrix $H={\left[\begin{matrix}H_{ij}\end{matrix}\right]}_{n_c\times n_v}$. Let $W$ be a TBC walk of length $2l$ in the protograph $G_B$ having edge sequence $e_1, e_2,\dots,e_{2l}$ with associated permutation matrices $P_1, P_2,\dots,P_{2l}$. Then the Tanner graph has a $2l$-cycle if the product $P_1P_2^\top P_3P_4^\top\cdots P_{2l-1}P_{2l}^\top$ has a fixed point.
\end{theorem}
In the process of choosing the permutation matrices to construct a parity-check matrix $H$ with target girth $g$, the condition in Theorem \ref{cycles_and_fixed_points_scalar} must be verified for every TBC walk of length $2l<g$ in the protograph. If the lifting factor $N$ is large, computing the product $P_1P_2^\top P_3P_4^\top\cdots P_{2l-1}P_{2l}^\top$ over all TBC walks becomes computationally prohibitive. To reduce the computational demand,  
we use the fact that this product of permutation matrices corresponds to the composition of their corresponding permutation polynomials, and we obtain the following immediate result. 
\begin{corollary}
\label{cycles_and_fixed_points_pp}
Let $\mathcal{C}$ be a protograph-based code with parity-check matrix $H={\left[\begin{matrix}H_{ij}\end{matrix}\right]}_{n_c\times n_v}$. Let $W$ be a TBC walk of length $2l$ in the protograph $G_B$ having edge sequence $e_1, e_2,\dots,e_{2l}$ with associated permutation matrices $P_1, P_2,\dots,P_{2l}$. If the corresponding permutation polynomials are $p_1, p_2,\dots,p_{2l}$, then the Tanner graph has a $2l$-cycle if the composition $p_{2l}^{-1} \circ p_{2l-1} \circ \cdots \circ p_4^{-1} \circ p_3 \circ p_2^{-1} \circ p_1$ has a fixed point.
\end{corollary}


\section{Beyond Affine Permutation Polynomials}
\label{sec:pp_beyond_affine}
In this paper, we consider two families of permutation polynomials that are closed under composition and inverses, which facilitates girth and cycle analysis.

\subsection{Monomials and Girth and Minimum Distance Upper Bounds}
Let $\mathbb{F}_q$ be the finite field with $q$ elements. The polynomial $f(x)=x^n$ is a permutation polynomial of $\mathbb{F}_q$ if and only if $\gcd(n,q-1)=1$. In this section, we propose the use of these permutation polynomials to construct protograph-based codes. We consider the $n_c\times n_v$ all-ones protograph for space constraints. Let $\mathcal{C}$ be a protograph-based LDPC code with parity-check matrix $H={\left[\begin{matrix} x^{l_{ij}}\end{matrix}\right]}_{n_c\times n_v}$, where $l_{ij} \in [q]$ satisfies $\gcd(l_{ij},q-1)=1$ for $i\in[n_c],j\in[n_v]$.

As with QC-LDPC codes, the commutativity of the permutation matrices obtained from the permutation polynomials of the form $f(x)=x^n$ forces upper bounds on both the girth of the Tanner graph and the minimum distance of the code $\mathcal{C}$ \cite{fossorier04,md01,sv12}.
\begin{theorem}
\label{girth_upper_bound_12_powers_of_x}
Let $\mathcal{C}$ be a protograph-based code with parity-check matrix $H$ with monomials. 
If $n_c\geq2$ and $n_v\geq3$, then $\text{girth}(H)\leq12$.
\end{theorem}
\begin{proof}
Consider any $2\times3$ submatrix $[x^a, x^b, x^c; x^d, x^e, x^f]$
of $H$ and let $W=[e_a, \allowbreak e_d,\allowbreak e_f, \allowbreak e_c, \allowbreak e_b, \allowbreak e_e, \allowbreak e_d, \allowbreak e_a, \allowbreak e_c, \allowbreak e_f, \allowbreak e_e, \allowbreak e_b]$, where $e_k$ is the edge associated to $x^k$. Then $W$ is a TBC walk in the protograph with 
$x^{a d^{-1} f c^{-1} b e^{-1} d a^{-1} c f^{-1} e b^{-1}}=x.$
Since $x$ has fixed points, 12-cycles are unavoidable.
\end{proof}
\begin{theorem}
\label{minimum_distance_upper_bound_powers_of_x}
Let $\mathcal{C}$ be a protograph-based code with parity-check matrix $H$ with monomials. 
Then the minimum distance of $\mathcal{C}$ is upper bounded by $(n_c+1)!$.
\end{theorem}
\begin{proof}
The permutation matrices obtained from permutation polynomials of the form $f(x)=x^n$ are commutative, so Theorem 2 in \cite{md01} applies. 
\end{proof}
Theorems \ref{girth_upper_bound_12_powers_of_x} and \ref{minimum_distance_upper_bound_powers_of_x} clearly show that the commutativity of the permutation matrices obtained from permutation polynomials of the form $f(x)=x^n$ is limiting. 
In addition, such monomials have typically several fixed points which further deteriorate their performance.
To prevent these upper bounds from holding, notice that it is enough to impose the condition that one of the permutation matrices in each $2\times3$ submatrix of the polynomial parity-check matrix $H$, like that in the proof of Theorem \ref{girth_upper_bound_12_powers_of_x}, does not commute with the others. 

\subsection{Permutation Polynomials $f(x)=ax^{(q+1)/2}+bx$}
\label{subsec:pp_binomials}

Let $\mathbb{F}_q$ be the finite field with $q$ elements, $q$ odd. 
We can partition $\mathbb{F}_q$ as $\mathbb{F}_q=\{0\}\cup S\cup N$ with 
$S=\{x\in\mathbb{F}_q^\times\mid \exists a \in \mathbb{F}_q^\times, x = a^2\}$ 
and $N=\{x\in \mathbb{F}_q^\times\mid \forall a\in \mathbb{F}_q, x \neq a^2\}$. 
Now consider the family of polynomials of the form $f(x)=ax^{(q+1)/2}+bx \in \mathbb{F}_q[x]$. This is a family of non-commuting polynomials that is closed under composition and inverses. These polynomials have been studied in the literature and conditions on their coefficients to make them permutation polynomials are known \cite{c62,nr82}. In this section, we use these polynomials to construct protograph-based codes. First, we lay out some structural results.

\begin{theorem}[\hspace{-1sp}\cite{nr82}]
\label{composition_of_binomials}
Let $q$ be odd and consider the polynomials $f_1(x)=ax^{(q+1)/2}+bx,f_2(x)=cx^{(q+1)/2}+dx\in\mathbb{F}_q[x]$. Then the composition $(f_2\circ f_1)(x)$ is given by
\begin{equation*}
(ce+da)x^{(q+1)/2}+(cf+db)x\mod(x^q-x)
,\end{equation*}
where $e+f={(a+b)}^{(q+1)/2}$ and $e+f={(b-a)}^{(q+1)/2}$.
\end{theorem}
\begin{theorem}[\hspace{-1sp}\cite{nr82}]
\label{binomials_permutation_polynomials}
The polynomial $f(x)=ax^{(q+1)/2}+bx$ is a permutation polynomial of $\mathbb{F}_q$ if and only if $b^2-a^2 \in S$.
\end{theorem}
\begin{prop}
The permutation polynomial $f(x)=ax^{(q+1)/2}+bx$ is piecewise affine. 
\end{prop}
\begin{proof}
Consider the quadratic character $\psi : \mathbb{F}_q \rightarrow \{0,1,-1\}$ defined by $\psi(x)=x^{(q-1)/2}$. Then
\vspace{-2mm}
\begin{equation*}
\psi(x) = 
\begin{cases} 
0, & \text{if } x=0, \\
1, & \text{if } x \in S, \\
-1, & \text{if } x \in N.
\end{cases}
\end{equation*}
We can rewrite the polynomial $f(x)=ax^{(q+1)/2}+bx$ as
\begin{align*}
f(x) &= ax^{(q+1)/2}+bx = (ax^{(q-1)/2}+b)x =(a\psi(x)+b)x,
\end{align*}
and by substituting the values for $\psi(x)$, we get
\begin{equation}
\label{pw-fn}
f(x) = 
\begin{cases} 
0, & \text{if } x=0, \\
f_{+}(x):=(b+a)x, & \text{if } x \in S, \\
f_{-}(x):=(b-a)x, & \text{if } x \in N.
\end{cases}
\end{equation}
This concludes the proof.
\end{proof}
\begin{theorem}
The permutation polynomial $f(x)=ax^{(q+1)/2}+bx$ has $x=0$ as its only fixed point if and only if $b\pm a\neq 1$.
\end{theorem}
\begin{proof}
This is clear from \eqref{pw-fn}.
\end{proof}

In order to determine the girth conditions using Theorem \ref{cycles_and_fixed_points_pp}, we are interested in analyzing how the compositions of $n$ of these polynomials look like. Here, we use the fact that if $a,b\in S$ and $c,d\in N$, then $ab\in S,\ cd\in S,\ ac\in N$. Consider the following two cases:

\textbf{Case 1:} Assume that $b_i\pm a_i\in S$. If $x\in S$, then $f_i(x)=(b_i+a_i)x\in S$, and if $x\in N$, then $f_i(x)=(b_i-a_i)x\in N$. Hence, the composition $F=f_n\circ\cdots\circ f_2\circ f_1$ is given by 
\begin{equation*}
F(x) = 
\begin{cases}
0, & \text{if } x=0,  \\
M_Sx=\prod_{i=1}^n(b_i+a_i)x, & \text{if } x \in S, \\
M_Nx=\prod_{i=1}^n(b_i-a_i)x, & \text{if } x \in N.
\end{cases}
\end{equation*}

\textbf{Case 2:} Assume that $b_i\pm a_i\in N$. If $x\in S$, then $f_i(x)=(b_i+a_i)x\in N$, and if $x\in N$, then $f_i(x)=(b_i-a_i)x\in S$. Hence, the composition $F=f_n\circ\cdots\circ f_2\circ f_1$ is given by 
\begin{equation*}
F(x) = 
\begin{cases}
0, & \text{if } x=0,  \\
M_S'x=\prod_{i=1}^n\left(b_i+{(-1)}^{i+1}a_i\right)x, & \text{if } x \in S, \\
M_N'x=\prod_{i=1}^n\left(b_i+{(-1)}^{i}a_i\right)x, & \text{if } x \in N.
\end{cases}
\end{equation*}

We obtain the following immediate consequence.

\begin{theorem}
\label{theorem_condition_only_0_fixed_point}
The composition $F=f_n\circ\cdots\circ f_2\circ f_1$ with permutation polynomials $f_i(x)=a_ix^{(q+1)/2}+b_ix$ has $x=0$ as its only fixed point if and only if $M_S,M_N,M_S',M_N'\neq 1$.
\end{theorem}

The general case where the composition $F$ has polynomials with both $b_i\pm a_i\in S$ and $b_j\pm a_j\in N$ is omitted due to space constraints. 

\subsection{Girth Upper Bounds}\label{subsec:girth_minimum_distance_upper_bounds_binomials}
Let $q$ be odd. Let $\mathcal{C}$ be a protograph-based LDPC code with parity-check matrix 
$H={\left[\begin{matrix} 
a_{ij}x^{(q+1)/2}+b_{ij}x
\end{matrix}\right]}_{n_c\times n_v}$,
where $b_{ij}^2-a_{ij}^2 \in S$ for $i\in[n_c],j\in[n_v]$. Similar to the construction using monomials, here we remove the rows and columns corresponding to the fixed point 0.
\begin{definition}
\label{definition_type_square_nonsquare}
For $f(x)=ax^{(q+1)/2}+bx$ with $b^2-a^2\in S$, we say that $f$ (and its corresponding permutation matrix) is of type $\bm{\mathsf{S}}$ if $b\pm a\in S$, and of type $\bm{\mathsf{N}}$ if $b\pm a\in N$.
\end{definition}
Assume that $n_c\geq2$. To compute girth upper bounds on these protograph-based LDPC codes, 
first notice that if $f_i$ and $f_j$ have corresponding permutation matrices $F_i$ and $F_j$, then there are four possible options for the type of the columns of a $2\times L$ submatrix of $H$: 
\begin{equation*}
\left[\begin{matrix}F_i \\ F_j\end{matrix}\right] \in \left\{ \left[\begin{matrix}\bm{\mathsf{S}} \\ \bm{\mathsf{S}}\end{matrix}\right], \left[\begin{matrix}\bm{\mathsf{S}} \\ \bm{\mathsf{N}}\end{matrix}\right], \left[\begin{matrix}\bm{\mathsf{N}} \\ \bm{\mathsf{S}}\end{matrix}\right], \left[\begin{matrix}\bm{\mathsf{N}} \\ \bm{\mathsf{N}}\end{matrix}\right] \right\}.
\end{equation*}
We will analyze how these four possible options affect the composition of these polynomials in Theorem \ref{cycles_and_fixed_points_pp}. It is important to mention that a polynomial $f$ of type $\bm{\mathsf{S}}$ sends squares to squares and non-squares to non-squares, so the restrictions to $S$ and $N$ are $f|_S:S\rightarrow S$ and $f|_N:N\rightarrow N$, respectively. On the other hand, a polynomial $f$ of type $\bm{\mathsf{N}}$ sends squares to non-squares and non-squares to squares, so the restrictions to $S$ and $N$ are $f|_S:S\rightarrow N$ and $f|_N:N\rightarrow S$, respectively. Consider the composition $f_j^{-1}\circ f_i$ and $x\in\mathbb{F}_q^\times$.
\begin{itemize}
\item \textbf{Case 1:} If $x\in S$, then $f_i=f_{i+}$ and $f_{i+}(x)\in S$, which implies that $f_j^{-1}=f_{j+}^{-1}$ and $f_{j+}^{-1}(f_{i+}(x))\in S$. If $x\in N$, then $f_i=f_{i-}$ and $f_{i-}(x)\in N$, which implies that $f_j^{-1}=f_{j-}^{-1}$ and $f_{j-}^{-1}(f_{i-}(x))\in N$.

\item \textbf{Case 2:} If $x\in S$, then $f_i=f_{i+}$ and $f_{i+}(x)\in S$, which implies that $f_j^{-1}=f_{j-}^{-1}$ and $f_{j-}^{-1}(f_{i+}(x))\in N$. If $x\in N$, then $f_i=f_{i-}$ and $f_{i-}(x)\in N$, which implies that $f_j^{-1}=f_{j+}^{-1}$ and $f_{j+}^{-1}(f_{i-}(x))\in S$.

\item \textbf{Case 3:} If $x\in S$, then $f_i=f_{i+}$ and $f_{i+}(x)\in N$, which implies that $f_j^{-1}=f_{j-}^{-1}$ and $f_{j-}^{-1}(f_{i+}(x))\in N$. If $x\in N$, then $f_i=f_{i-}$ and $f_{i-}(x)\in S$, which implies that $f_j^{-1}=f_{j+}^{-1}$ and $f_{j+}^{-1}(f_{i-}(x))\in S$.

\item \textbf{Case 4:} If $x\in S$, then $f_i=f_{i+}$ and $f_{i+}(x)\in N$, which implies that $f_j^{-1}=f_{j+}^{-1}$ and $f_{j+}^{-1}(f_{i+}(x))\in S$. If $x\in N$, then $f_i=f_{i-}$ and $f_{i-}(x)\in S$, which implies that $f_j^{-1}=f_{j-}^{-1}$ and $f_{j-}^{-1}(f_{i-}(x))\in N$.
\end{itemize}
From this discussion, notice that we can divide these four cases in two categories. Category 1 consists of Cases 2 and 3, which are characterized by having a composition $f_j^{-1}\circ f_i$ 
that sends squares to non-squares and vice versa. 
On the other hand, Category 2 consists of Cases 1 and 4, which are characterized by having a composition $f_j^{-1}\circ f_i$ 
that sends squares to squares and non-squares to non-squares.

\begin{theorem}
\label{girth_upper_bound_binomials}
Let $\mathcal{C}$ be a protograph-based code with parity-check matrix $H$ using $f(x)=ax^{(q+1)/2}+bx$. 
Then we have the following.
\begin{enumerate}[label=(\alph*)]
    \item If the protograph has a $2\times3$ all-ones submatrix, then $\text{girth}(H)\leq20$.
    \item If the protograph has a $2\times4$ all-ones submatrix, then $\text{girth}(H)\leq16$.
    \item If the protograph has a $2\times L$ all-ones submatrix with $L\geq5$, then $\text{girth}(H)\leq12$.
\end{enumerate}
\end{theorem}
\begin{proof}
For $L\geq3$, consider any $2\times L$ submatrix
\begin{equation*}
B_{2\times L}=\left[\begin{matrix}f_1&f_2&f_3&\cdots&f_L \\ f_{L+1}&f_{L+2}&f_{L+3}&\cdots&f_{2L}\end{matrix}\right]
\end{equation*}
of $H$, 
where $f_i=a_ix^{(q+1)/2}+b_ix$ satisfies $b_i^2-a_i^2\in S$ for each $i$.

(a) Take all the columns of $B_{2\times3}$ to be in Category 1, say
\begin{equation*}
B_{2\times 3 } = \left[\begin{matrix}f_1&f_2&f_3 \\ f_{4}&f_{5}&f_{6}\end{matrix}\right] \gets \left[\begin{matrix}\bm{\mathsf{N}}&\bm{\mathsf{S}}&\bm{\mathsf{N}} \\ \bm{\mathsf{S}}&\bm{\mathsf{N}}&\bm{\mathsf{S}}\end{matrix}\right],
\end{equation*}
and let $W=[e_1,e_4,e_5,e_2,e_3,e_6,e_4,e_1,e_2,e_5,e_6,e_3]$, where $e_i$ is the edge associated to $f_i$. Then $W$ is a TBC walk of length 12 in the protograph. If we take $x=1$, then the composition of the polynomials following Theorem \ref{cycles_and_fixed_points_pp}, in the order prescribed by $W$, gives
\begin{align*}
f_{3}^{-1}&f_{6}f_{5}^{-1}f_{2}f_{1}^{-1}f_{4}f_{6}^{-1}f_{3}f_{2}^{-1}f_{5}f_{4}^{-1}f_{1}(1)= \\
&f_{3+}^{-1}f_{6-}f_{5-}^{-1}f_{2+}f_{1+}^{-1}f_{4-}f_{6-}^{-1}f_{3+}f_{2+}^{-1}f_{5-}f_{4-}^{-1}f_{1+}(1)=1.
\end{align*}
Hence, there is a 12-cycle in the Tanner graph. The analysis for any other combination of three columns in Category 1 is similar and omitted for space constraints. 

Now take all the columns of $B_{2\times3}$ to be in Category 2, say
\begin{equation*}
B_{2\times 3 } = \left[\begin{matrix}f_1&f_2&f_3 \\ f_{4}&f_{5}&f_{6}\end{matrix}\right] \gets \left[\begin{matrix}\bm{\mathsf{N}}&\bm{\mathsf{S}}&\bm{\mathsf{N}} \\ \bm{\mathsf{N}}&\bm{\mathsf{S}}&\bm{\mathsf{N}}\end{matrix}\right],
\end{equation*}
and let $W=[e_1,e_4,e_5,e_2,e_3,e_6,e_4,e_1,e_2,e_5,e_6,e_3]$, which is a TBC walk of length 12 in the protograph. If we take $x=1$, then the corresponding composition gives
\begin{align*}
f_{3}^{-1}&f_{6}f_{5}^{-1}f_{2}f_{1}^{-1}f_{4}f_{6}^{-1}f_{3}f_{2}^{-1}f_{5}f_{4}^{-1}f_{1}(1)= \\
&f_{3+}^{-1}f_{6+}f_{5+}^{-1}f_{2+}f_{1+}^{-1}f_{4+}f_{6+}^{-1}f_{3+}f_{2+}^{-1}f_{5+}f_{4+}^{-1}f_{1+}(1)=1.
\end{align*}
Hence, there is a 12-cycle in the Tanner graph. The analysis for any other combination of three columns in Category 2 is similar and omitted for space constraints. 

Now, assume that $B_{2\times3}$ has columns in both Category 1 and Category 2, say
\begin{equation*}
B_{2\times 3 } = \left[\begin{matrix}f_1&f_2&f_3 \\ f_{4}&f_{5}&f_{6}\end{matrix}\right] \gets \left[\begin{matrix}\bm{\mathsf{S}}&\bm{\mathsf{S}}&\bm{\mathsf{N}} \\ \bm{\mathsf{S}}&\bm{\mathsf{N}}&\bm{\mathsf{S}}\end{matrix}\right],
\end{equation*}
and let $W=[e_2,\allowbreak e_5,\allowbreak e_4,\allowbreak e_1,\allowbreak e_2,\allowbreak e_5,\allowbreak e_6,\allowbreak e_3,\allowbreak e_1,\allowbreak e_4,\allowbreak e_5,\allowbreak e_2,\allowbreak e_1,\allowbreak e_4,\allowbreak e_5,\allowbreak e_2,\allowbreak e_3,\allowbreak e_6,\allowbreak e_4,\allowbreak e_1]$, which is a TBC walk of length 20 in the protograph. If we take $x=1$, then the corresponding composition gives
\begin{align*}
	&\begin{aligned}
    f_{1}^{-1}f_{4}f_{6}^{-1}f_{3}f_{2}^{-1}f_{5}&f_{4}^{-1}f_{1}f_{2}^{-1} \cdots \\ &\cdots f_{5}f_{4}^{-1}f_{1}f_{3}^{-1}f_{6}f_{5}^{-1}f_{2}f_{1}^{-1}f_{4}f_{5}^{-1}f_{2}(1)
	\end{aligned} \\
	&\begin{aligned}
	=f_{1+}^{-1}&f_{4+}f_{6+}^{-1}f_{3-}f_{2-}^{-1}f_{5+}f_{4+}^{-1}f_{1+}f_{2+}^{-1}f_{5-}\cdots \\ &\cdots f_{4-}^{-1}f_{1-}f_{3-}^{-1}f_{6+}f_{5+}^{-1}f_{2-}f_{1-}^{-1}f_{4-}f_{5-}^{-1}f_{2+}(1)=1.
	\end{aligned}
\end{align*}
Hence, there is a 20-cycle in the Tanner graph. The analysis for any other combination of two columns in either category and one column in the other category is similar and omitted for space constraints. 

(b) Similar to part (a) and omitted for space constraints.

(c) Three columns from one category are always found.
\end{proof}

Since the composition of permutation polynomials of the form $f(x)=ax^{(q+1)/2}+bx$ does not commute in general, the corresponding permutation matrices do not commute, which means that the upper bound on the minimum distance for QC-LDPC codes does not apply. It is future work to investigate this upper bound.

\vspace{-2mm}
\section{Simulations}
\label{sec:simulations}
In this section, we compare the performance of protograph-based codes proposed in this paper---referred generically as PM-LDPC---with APM-LDPC codes
under the belief propagation decoding algorithm \cite{fmi99}. As a baseline, we also include an instance of QC-LDPC codes.

Consider the polynomial parity-check matrices $F_{2\times3}=\left[\begin{matrix}f_{ij}(x)\end{matrix}\right]$, $G_{2\times3}=\left[\begin{matrix}g_{ij}(x)\end{matrix}\right]$, and $H_{2\times3}=\left[\begin{matrix}h_{ij}(x)\end{matrix}\right]$, with $f_{ij}(x)\in\mathbb{F}_{101}[x]$ and $g_{ij}(x),h_{ij}(x)\in\mathbb{Z}_{100}[x]$, for $i\in[2],j\in[3]$, as in Table \ref{table_ex_2x3_constructions}. 
For lifting factor $N=q-1=100$, both $F_{2\times3}$ and $G_{2\times3}$ have girth 20 with 200 and 350 20-cycles, respectively, and $H_{2\times3}$ has girth 12 with 100 12-cycles. The parameters of the linear codes obtained from these constructions are $[300,101,10]$ for both $F_{2\times3}$ and $G_{2\times3}$, and $[300,101,6]$ for $H_{2\times3}$. 
\begin{table}[t]
\caption{PMs, APMs, and CPMs used to construct $F_{2\times3}$, $G_{2\times3}$, and $H_{2\times3}$.}
\label{table_ex_2x3_constructions}
\centering
\begin{tabular}{||c|c|c|c||}
\hline 
\multicolumn{1}{||c||}{} & PMs & APMs & CPMs \\ \hline \hline
\multicolumn{1}{||c||}{$(i,j)$} & $f_{ij}(x)$ & $g_{ij}(x)$ & $h_{ij}(x)$ \\ \hline \hline
\multicolumn{1}{||c||}{$(0,0)$} & $52x^{51}+87x$ & $x$ & $x+81$  \\ \hline
\multicolumn{1}{||c||}{$(0,1)$} & $68x^{51}+36x$ & $x$ & $x+14$  \\ \hline
\multicolumn{1}{||c||}{$(0,2)$} & $74x^{51}+42x$ & $x$ & $x+3$  \\ \hline
\multicolumn{1}{||c||}{$(1,0)$} & $40x^{51}+58x$ & $x$ & $x+94$  \\ \hline
\multicolumn{1}{||c||}{$(1,1)$} & $82x^{51}+49x$ & $x+5$ & $x+35$  \\ \hline
\multicolumn{1}{||c||}{$(1,2)$} & $62x^{51}+11x$ & $11x+1$ & $x+31$  \\
\hline 
\end{tabular}
\end{table}
We define the matrices $F_{2\times4}$, $G_{2\times4}$, and $H_{2\times4}$ in a similar way 
using Table \ref{table_ex_2x4_constructions}.
For lifting factor $N=q-1=100$, both $F_{2\times4}$ and $G_{2\times4}$ have girth 16 with 800 and 1350 16-cycles, respectively, and $H_{2\times4}$ has girth 12 with 400 12-cycles. The parameters of the linear codes obtained from these constructions are $[400,201,8]$ for both $F_{2\times4}$ and $G_{2\times4}$, and $[400,201,6]$ for $H_{2\times4}$. 
As seen in Figure~\ref{fig:simulations}, PM-LDPC codes yield a slight decoding gain. Although not discussed here in detail due to space limitations, the same decoding gain can be achieved with smaller lift sizes.

\begin{table}[t]
\caption{PMs, APMs, and CPMs used to construct $F_{2\times4}$, $G_{2\times4}$, and $H_{2\times4}$.} 
\label{table_ex_2x4_constructions}
\centering
\begin{tabular}{||c|c|c|c||}
\hline 
\multicolumn{1}{||c||}{} & PMs & APMs & CPMs \\ \hline \hline
\multicolumn{1}{||c||}{$(i,j)$} & $f_{ij}(x)$ & $g_{ij}(x)$ & $h_{ij}(x)$ \\ \hline \hline
\multicolumn{1}{||c||}{$(0,0)$} & $57x^{51}+15x$ & $x$ & $x+81$  \\ \hline
\multicolumn{1}{||c||}{$(0,1)$} & $88x^{51}+8x$ & $x$ & $x+14$  \\ \hline
\multicolumn{1}{||c||}{$(0,2)$} & $62x^{51}+x$ & $x$ & $x+3$  \\ \hline
\multicolumn{1}{||c||}{$(0,3)$} & $31x^{51}+56x$ & $x$ & $x+94$  \\ \hline
\multicolumn{1}{||c||}{$(1,0)$} & $75x^{51}+93x$ & $x$ & $x+35$  \\ \hline
\multicolumn{1}{||c||}{$(1,1)$} & $31x^{51}+45x$ & $69x+87$ & $x+31$  \\ \hline
\multicolumn{1}{||c||}{$(1,2)$} & $60x^{51}+8x$ & $69x+73$ & $x+28$  \\ \hline
\multicolumn{1}{||c||}{$(1,3)$} & $83x^{51}+44x$ & $x+34$ & $x+17$  \\
\hline 
\end{tabular}
\vspace{-15pt}
\end{table}

\begin{figure}[t]
\begin{center}
\includegraphics[width=0.80\columnwidth]{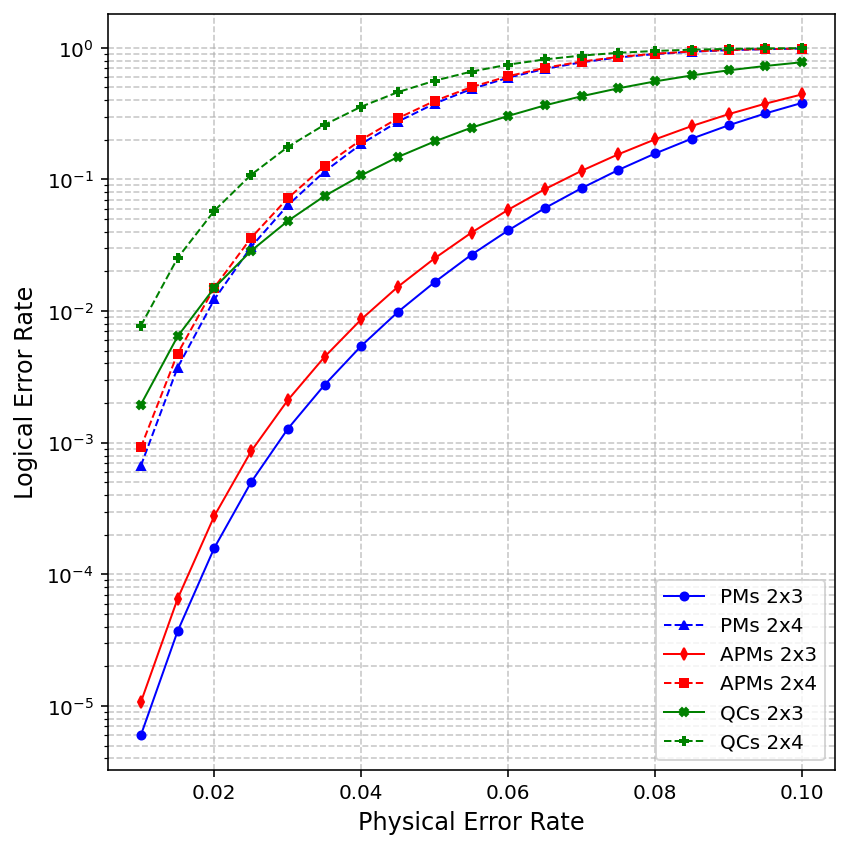}\vspace{-5mm}
\end{center}
\caption{Simulated decoding performance for the codes constructed using the binomial permutation polynomials (PMs), affine permutation polynomials (APMs), and circulant permutation matrices (CPMs) in Tables \ref{table_ex_2x3_constructions} and \ref{table_ex_2x4_constructions}.} 
\label{fig:simulations}\vspace{-5mm}
\end{figure}

\section{Concluding Remarks}
\label{sec:conclusion}
In this paper, we propose constructing protograph-based LDPC codes from permutation polynomials over finite fields, generalizing the notion of QC-LDPC and APM-LDPC codes.
We focus on two simple classes of permutation polynomials, namely permutation monomials $f(x) = x^n$ and permutation binomials $f(x)=ax^{(q+1)/2}+bx$ in $\mathbb{F}_q[x]$.
Both families are closed under composition and inverses, which in turn enables an algebraic and efficient girth analysis.
While an interesting case study, permutation monomials commute and typically have several fixed points, and this yields a deteriorated performance.
On the other hand, the permutation binomials do not commute. In fact, we show that they are piecewise affine; see Proposition~\ref{binomials_permutation_polynomials}.
Simulations suggest that protograph-based LDPC codes constructed with said permutation binomials have a lower error rate than APM-LDPC codes.
In addition, the same parameters can be obtained with a lower lift size. For instance, using permutation binomials, girth 20 can be achieved with lift size 18 (that is, binomials of $\mathbb{F}_{19}$ and eliminating the fixed point $x=0$). In contrast, for affine permutations, to achieve girth 20 one needs lift size 20.

Based on the simplicity of the permutation polynomials considered in this paper, we are cautiously optimistic that more sophisticated permutation polynomials will yield meaningful coding gains. 
In future work, we will explore other families of permutation polynomials that are closed under composition and inverses such as the Dickson polynomials~\cite{d1896}.

Another interesting future direction is the exploration of quantum LDPC codes.
Kasai~\cite{k25} has already demonstrated that affine permutations can be used for the construction of high-girth quantum LDPC codes.
In addition, the work of~\cite{komoto-kasai} lays out a decoding algorithm with linear (in the number of physical qubits) computational cost, making them appealing for large-scale fault-tolerant quantum computing.
Again, we are cautiously optimistic that these works can be further improved using general permutation polynomials.

\balance
\bibliographystyle{IEEEtran}
\bibliography{References_May_2026.bib}

\begin{thebibliography}{10}
\providecommand{\url}[1]{#1}
\csname url@rmstyle\endcsname
\providecommand{\newblock}{\relax}
\providecommand{\bibinfo}[2]{#2}
\providecommand\BIBentrySTDinterwordspacing{\spaceskip=0pt\relax}
\providecommand\BIBentryALTinterwordstretchfactor{4}
\providecommand\BIBentryALTinterwordspacing{\spaceskip=\fontdimen2\font plus
\BIBentryALTinterwordstretchfactor\fontdimen3\font minus
  \fontdimen4\font\relax}
\providecommand\BIBforeignlanguage[2]{{%
\expandafter\ifx\csname l@#1\endcsname\relax
\typeout{** WARNING: IEEEtran.bst: No hyphenation pattern has been}%
\typeout{** loaded for the language `#1'. Using the pattern for}%
\typeout{** the default language instead.}%
\else
\language=\csname l@#1\endcsname
\fi
#2}}

\bibitem{lcz+06}
Z.~Li, L.~Chen, L.~Zeng, S.~Lin, and W.~H. Fong, ``Efficient encoding of
  quasi-cyclic low-density parity-check codes,'' \emph{IEEE Transactions on
  Communications}, vol.~54, no.~1, pp. 71--81, January 2006.

\bibitem{wc07}
Z.~Wang and Z.~Cui, ``A memory efficient partially parallel decoder
  architecture for quasi-cyclic {LDPC} codes,'' \emph{IEEE Transactions on Very
  Large Scale Integration (VLSI) Systems}, vol.~15, no.~4, pp. 483--488, April
  2007.

\bibitem{md01}
D.~J.~C. MacKay and M.~C. Davey, ``Evaluation of {Gallager} codes for short
  block length and high rate applications,'' in \emph{Codes, Systems, and
  Graphical Models}, B.~Marcus and J.~Rosenthal, Eds., vol. 123.\hskip 1em plus
  0.5em minus 0.4em\relax Springer New York, 2001, pp. 113--130.

\bibitem{fossorier04}
M.~P.~C. Fossorier, ``Quasi-cyclic low-density parity-check codes from
  circulant permutation matrices,'' \emph{IEEE Transactions on Information
  Theory}, vol.~50, no.~8, pp. 1788--1793, August 2004.

\bibitem{sv12}
R.~Smarandache and P.~O. Vontobel, ``Quasi-cyclic {LDPC} codes: Influence of
  proto- and {Tanner}-graph structure on minimum {Hamming} distance upper
  bounds,'' \emph{IEEE Transactions on Information Theory}, vol.~58, no.~2, pp.
  585--607, February 2012.

\bibitem{mb01}
Y.~Mao and A.~H. Banihashemi, ``A heuristic search for good low-density
  parity-check codes at short block lengths,'' \emph{IEEE International
  Conference on Communications Proceedings}, vol.~1, pp. 41--44, June 2001.

\bibitem{hc06}
T.~R. Halford and K.~M. Chugg, ``An algorithm for counting short cycles in
  bipartite graphs,'' \emph{IEEE Transactions on Information Theory}, vol.~52,
  no.~1, pp. 287--292, January 2006.

\bibitem{kb12_b}
M.~Karimi and A.~H. Banihashemi, ``Efficient algorithm for finding dominant
  trapping sets of {LDPC} codes,'' \emph{IEEE Transactions on Information
  Theory}, vol.~58, no.~11, pp. 6942--6958, November 2012.

\bibitem{ru01}
T.~Richardson and R.~Urbanke, ``Efficient encoding of low-density parity-check
  codes,'' \emph{IEEE Transactions on Information Theory}, vol.~47, no.~2, pp.
  638--656, February 2001.

\bibitem{h1863}
C.~Hermite, ``Sur les fonctions de sept lettres,'' \emph{Comptes Rendus de
  l'Acad\'emie des Sciences de Paris}, vol.~57, no. 750--757, 1863.

\bibitem{d1896}
L.~E. Dickson, ``The analytic representation of substitutions on a power of a
  prime number of letters with a discussion of the linear group,'' \emph{Annals
  of Mathematics}, vol.~11, no. 1/6, pp. 65--120, 1896.

\bibitem{d13}
C.~Ding, ``Cyclic codes from some monomials and trinomials,'' \emph{SIAM
  Journal on Discrete Mathematics}, vol.~27, no.~4, pp. 1977--1994, November
  2013.

\bibitem{dy06}
C.~Ding and J.~Yuan, ``A family of skew hadamard difference sets,''
  \emph{Journal of Combinatorial Theory, Series A}, vol. 113, no.~7, pp.
  1526--1535, October 2006.

\bibitem{lm84}
\BIBentryALTinterwordspacing
R.~Lidl and W.~B. M{\"u}ller, \emph{Permutation Polynomials in
  RSA-Cryptosystems}.\hskip 1em plus 0.5em minus 0.4em\relax Boston, MA:
  Springer US, 1984, pp. 293--301. [Online]. Available:
  \url{https://doi.org/10.1007/978-1-4684-4730-9_23}
\BIBentrySTDinterwordspacing

\bibitem{ccz98}
C.~Carlet, P.~Charpin, and V.~Zinoviev, ``Codes, bent functions and
  permutations suitable for des-like cryptosystems,'' \emph{Designs, Codes and
  Cryptography}, vol.~15, no.~2, pp. 125--156, November 1998.

\bibitem{d99a}
H.~Dobbertin, ``Almost perfect nonlinear power functions on {GF}(2n): The niho
  case,'' \emph{Information and Computation}, vol. 151, no.~1, pp. 57--72, May
  1999.

\bibitem{d99b}
------, ``Almost perfect nonlinear power functions on {GF}(2n): the welch
  case,'' \emph{IEEE Transactions on Information Theory}, vol.~45, no.~4, pp.
  1271--1275, May 1999.

\bibitem{myp06}
S.~Myung, K.~Yang, and D.~S. Park, ``A combining method of structured {LDPC}
  codes from affine permutation matrices,'' in \emph{2006 IEEE International
  Symposium on Information Theory}, July 2006, pp. 674--678.

\bibitem{yk19}
R.~Yoshida and K.~Kasai, ``Linear permutation polynomial codes,'' in \emph{2019
  IEEE International Symposium on Information Theory (ISIT)}, 2019, pp. 66--70.

\bibitem{gg24}
M.~Gholami and Z.~Gholami, ``An efficient algorithm for counting cycles in {QC}
  and {APM} {LDPC} codes,'' \emph{Facta Universitatis}, vol.~39, no.~4, pp.
  695--711, 2024.

\bibitem{k25}
K.~Kasai, ``Quantum error correction with girth-16 non-binary {LDPC} codes via
  affine permutation construction,'' in \emph{2025 13th International Symposium
  on Topics in Coding (ISTC)}, 2025, pp. 1--5.

\bibitem{tho03}
J.~Thorpe, ``Low-density parity-check ({LDPC}) codes constructed from
  protographs,'' \emph{Jet Propulsion Laboratory Pasadena, CA, INP Progress
  Report 42-154}, pp. 42--154, August 2003.

\bibitem{tan81}
R.~M. Tanner, ``A recursive approach to low complexity codes,'' \emph{IEEE
  Transactions on Information Theory}, vol.~27, no.~5, pp. 533--547, September
  1981.

\bibitem{fan00}
J.~L. Fan, ``Array codes as low-density parity-check codes,'' in \emph{2nd
  International Symposium on Turbo Codes and Related Topics}, September 2000,
  pp. 543--546.

\bibitem{c62}
L.~Carlitz, ``Some theorems on permutation polynomials,'' \emph{Bulletin of the
  American Mathematical Society}, vol.~68, no.~2, pp. 120--122, March 1962.

\bibitem{nr82}
H.~Niederreiter and K.~H. Robinson, ``Complete mappings of finite fields,''
  \emph{Journal of the Australian Mathematical Society}, vol.~33, no.~2, pp.
  197--212, October 1982.

\bibitem{fmi99}
M.~P.~C. Fossorier, M.~Mihaljevic, and H.~Imai, ``Reduced complexity iterative
  decoding of low-density parity check codes based on belief propagation,''
  \emph{IEEE Transactions on Communications}, vol.~47, no.~5, pp. 673--680, May
  1999.

\bibitem{komoto-kasai}
D.~Komoto and K.~Kasai, ``Quantum error correction near the coding theoretical
  bound,'' \emph{npj Quantum Information}, vol.~11, no. 154, 2025.

\end{thebibliography}

\end{document}